\documentclass[11pt]{article}
\def\nottoobig#1{{\hbox{$\left#1\vcenter to1.111\ht\strutbox{}\right.\n@space$}}}

\newtheorem{fact}{Fact}

\newtheorem{theorem}{Theorem}[section]

\newtheorem{lemma}[theorem]{Lemma}
\newtheorem{claim}[theorem]{Claim}

\newtheorem{definition}[theorem]{Definition}

\usepackage{epsfig}
\usepackage{amsmath}
\usepackage{amsfonts}

\newcommand{\nat}{{\mathbb N}}

\newcommand{\poly}{{\rm poly}}

\def\nottoobig#1{{\hbox{$\left#1\vcenter
to1.111\ht\strutbox{}\right.\n@space$}}}

\newcommand{\prob}{{\rm Prob}}

\newcommand{\ie}{$\mbox{i.e.}$}

\newlength{\filength}
\newsavebox{\gcbox}
\sbox{\gcbox}{\framebox[\filength]{\rule{0ex}{2ex}}}

\newcommand{\qedblob}{\mbox{\rule[-1.5pt]{5pt}{10.5pt}}}
\def\literalqed{{\ \nolinebreak\hfill\mbox{\qedblob\quad}}}

\def\qed{\literalqed}

\newcommand{\singlespacing}{\let\CS=
\@currsize\renewcommand{\baselinestretch}{1}\tiny\CS}
\newcommand{\singlespacingplus}{\let\CS=
\@currsize\renewcommand{\baselinestretch}{1.25}\tiny\CS}
\newcommand{\doublespacing}{\let\CS=
\@currsize\renewcommand{\baselinestretch}{1.75}\tiny\CS}
\newcommand{\draftspacing}{\let\CS=
\@currsize\renewcommand{\baselinestretch}{2.0}\tiny\CS}

\def\zo{\{0,1\}}

\def\mapping{\rightarrow}

\usepackage{amsfonts}

\if01
\fi

\if
\fi

\makeatletter
\def\@listI{\leftmargin\leftmargini \parsep 4.5pt plus 1pt minus 1pt\topsep6pt plus 2pt minus 2pt \itemsep  2pt plus 2pt minus 1pt}

\let\@listi\@listI
\@listi
\makeatother

\author{ {Marius Zimand\/}
\thanks{  Department of Computer and Information Sciences, Towson University,
Baltimore, MD.; email: mzimand@towson.edu; http://triton.towson.edu/\~{ }mzimand.
The author is supported in part
by NSF grant CCF 1016158.}}

\date{ }

\title{Symmetry of information and bounds on nonuniform randomness extraction via Kolmogorov extractors}

\begin{document}

\maketitle

\begin{abstract}
We prove a strong Symmetry of Information relation for random strings (in the sense of Kolmogorov complexity) and establish tight bounds on the amount on nonuniformity that is necessary for extracting a string with randomness rate $1$ from a single source of randomness. More precisely, as instantiations of more general results, we show:
\begin{itemize}
	\item For all $n$-bit random strings $x$ and $y$, $x$ is random conditioned by $y$ if and only if $y$ is random conditioned by $x$;
	\item While $O(1)$ amount of advice regarding the source is not enough for extracting a string with randomness rate $1$ from a source string with constant random rate, $\omega(1)$ amount of advice is.
\end{itemize}
The proofs use Kolmogorov extractors as the main technical device.

\end{abstract}

\smallskip

\section{Introduction}
Kolmogorov extractors are procedures that increase the Kolmogorov complexity rate of strings and sequences. Their explicit study was initiated by Fortnow, Hitchcock, A.~Pavan, Vinodchandran and Wang~\cite{fhpvw:c:extractKol} for the case of finite strings and by Reimann~\cite{rei:t:thesis} for the case of infinite strings. The recent paper~\cite{zim:j:kolmextractsurvey} is a survey of this field.

In this paper, we use Kolmogorov extractors as a conceptual and technical device to derive results in two directions that otherwise appear unrelated.

First, we study the symmetry of information phenomenon for random strings. We show that a random string $x$ has essentially no information about a random string $y$ if and only if $y$ has essentially no amount of information about $x$. By ``essentially no amount of information," we mean that the amount of information is bounded by an absolute constant.  Up to this paper, it was only known that if $x$ has only a constant amount information about $y$, then $y$ has $O(\log n)$ information about $x$. Thus we replace the $O(\log n)$ term by $O(1)$.

Secondly, we investigate the amount of non-uniformity that is necessary for randomness extraction from one source. It is well-known that randomness extraction from a single source is not possible (if we exclude some trivial cases). In fact, as a consequence of a result of Vereshchagin and Vyugin~\cite{ver-vyu:j:kolm}, we note that obtaining a source with randomness rate $1$ from a source with randomness rate, say, $0.99$ is not possible even if the extractor has access to a constant amount of non-uniform information. In contrast, we show that an $\omega(1)$ amount of non-uniform information is sufficient for this task.

We continue with a more detailed discussion of the two types of results and of the technical method that we use.
\smallskip

{\bf Symmetry of Information.} Strictly speaking, a string $x$ of length $n$ is said to be random if its Kolmogorov complexity is at least $n$ (\ie, $C(x \mid n) \geq n$) and is said to be random conditioned by string $y$ if its Kolmogorov complexity conditioned by $y$ is at least $n$ (\ie, $C(x \mid y) \geq n$). Since the Kolmogorov complexity depends up to an additive constant on the choice of the universal machine, a relaxed definition is more robust. Such a definition is obtained if we require that $C(x \mid n)$ (respectively, $C(x \mid y)$) are within a constant from $n$. Formally, for a constant $c$, we say that $x$ is $c$-random if $C(x \mid n) \geq n-c$, and $x$ is $c$-random conditioned by $y$ if $C(x \mid y) \geq n-c$.

We prove a strong symmetry-of-information relation for random strings: for any $n$-bit random strings $x$ and $y$, $x$ is random conditioned by $y$ iff $y$ is random conditioned by $x$.
More exactly, we show the following.
\smallskip

\begin{fact}
\label{t:symrandom}
(Informal statement; see Theorem~\ref{t:symrandstrings} for full statement.)
Let $x$ and $y$ be two $n$-bit strings that are $c$-random and such that $x$ is $c$-random conditioned by $y$. Then $y$ is $c'$-random conditioned by $x$, where $c'$ is a constant that only depends on $c$.  
\end{fact}
\smallskip

This result is a consequence of a new form of the Symmetry of Information Theorem, which is tighter than the classical theorem of Kolmogorov and Levin~\cite{zvo-lev:j:kol} for some types of strings, including the class of random strings. This is an important class because most strings are random and because counting arguments based on Kolmogorov complexity typically use random strings.
Symmetry of information is one of the basic principles in information theory. It states that for any two random variables $X$ and $Y$, the information in $X$ about $Y$ is \emph{equal} to the information in $Y$ about $X$, \ie,  $I(X : Y) = I(Y : X)$, where $I(X : Y) = H(Y) - 
H(Y \mid~X)$ and $H( )$ is the Shannon entropy.
The principle also holds in algorithmical information theory, provided we replace \emph{equal} by \emph{approximately equal}. More formally, for two binary strings $x$ and $y$, we define $I(x : y) = C(y) - C(y \mid x)$. Then, the Symmetry of Information Theorem of Kolmogorov and Levin states that
 $| I(x:y) - I(y:x) | = O(\log |x| + \log |y|)$. For the general case of arbitrary strings, the logarithmical term cannot be avoided because there are $n$-bit strings $x$ and $y$ such that 
$| I(x:y) - I(y:x) | > \log n - O(1)$ (for example, such strings can be obtained by a simple modification of Example 2.27 in~\cite{li-vit:b:kolmbook-first-ed}).

Symmetry of information follows immediately from another basic result, the \emph{Chain Rule}: $C(xy) \approx C(y) + C(x \mid y)$. While it is easy to see that $C(xy) \leq C(y) + C(x \mid~y) + C^{(2)} (y)$
\footnote{We use the following shortcut notations: $C^{(2)}(x)$ is $C(C(x))$, $C^{(2)}(x \mid~y)$ is $C(C(x \mid~y))$, $C^{(2)}(x \mid~n)$ is $C(C(x \mid~n) \mid n)$.},
the more difficult direction in the proof of the chain rule shows that $$C(xy) \geq C(y) + C(x \mid y) - C^{(2)} (xy).$$
We prove a new form of the Chain Rule (which implies the alluded new form of the Symmetry of Information Theorem; here and in the rest of this paper, $I(x : y) = C(y \mid n) - 
C(y \mid~x)$.)
\smallskip

\begin{fact} 
\label{t:chain}
(Informal statement;  see Theorem~\ref{t:lowerboundsym} for full statement)
Let $x$ and $y$ be $n$-bit strings such that $C(x \mid~y) = \Omega (\log n)$ and $C(y \mid~x) = \Omega (\log n)$. Then
\[
\begin{array}{ll}
 C(xy \mid n)  & \geq C(x \mid n) + C(y \mid x) - \log I(x : y)  \\
               & \quad \quad - O(C^{(2)}(x \mid n) + C^{(2)}(y \mid n) + C^{(2)}(y \mid x) ). 
\end{array}
\]
\end{fact}
\smallskip

Note that if $x$ is random and $y$ is random even conditioned by $x$, then $C^{(2)}(x \mid n) = O(1)$, $C^{(2)}(y \mid n) = O(1)$, $C^{(2)}(y \mid x) = O(1)$, and $I(x : y) = O(1)$. In this case we obtain
$ C(xy \mid n) \geq C(x \mid n) + C(y \mid x) - O(1) = 2n - O(1)$. Fact~\ref{t:symrandom} follows easily from this relation.
\medskip

{\bf Randomness extraction with small advice.} By and large, randomness extraction is an algorithmical process that constructs a source of randomness of high quality from one or several sources of lower quality. If we restrict to the case when there is only one input source, one wants to design an effective transformation $E$ from the set of $n$-bit strings to the set of $m$-bit strings such that for any source $x$ ``with randomness $\leq k$", $E(x)$ has randomness $\approx m$. It is desirable to have $m \approx k$ (\ie, to extract all, or almost all, of the randomness in the source). The problem of randomness extraction has been modeled in two ways. In the first model, a source is a probability distribution $X$ over $\zo^n$ and its randomness is given by the min-entropy $H_\infty(X)$. In the second model, a source is a string $x \in \zo^n$ and its randomness is given by its Kolmogorov complexity $C(x)$. For our study, the second model is more convenient; moreover all the results can be translated in the first model. (For the relation between the two models, see~\cite{fhpvw:c:extractKol}, ~\cite{hit-pav-vin:c:Kolmextraction} and also the survey paper~\cite{zim:j:kolmextractsurvey}).

Formally (in the second model), given the parameter $k \leq n$, one would like to have a function $E : \zo^n \mapping \zo^m$ such that $C(E(x)) \approx m$ whenever $C(x) \geq k$. It is well known that no such computable function $E$ exists for non-trivial parameters. Indeed for any given $E$, consider the string $y \in \zo^m$ with the largest number of preimages. Then $C(y \mid n) = O(1)$ and among its at least $2^{n-m}$ preimages there must be some $x$ with $C(x) \geq n-m$. In other words, for any given $E$, there are some strings (such as the above $x$) on which $E$ fails. Thus, in order for a function $E$ to extract randomness from any source $x$ with randomness $\geq k$, $E$ must have some additional information $\alpha_x$, which we call \emph{advice about the source.} The question is how much such advice information should be provided.

Fortnow et al.~\cite{fhpvw:c:extractKol} have shown that a constant number of advice bits are sufficient if one settles to extracting from strings with linear randomness a string whose randomness rate is $1-\epsilon$.~\footnote{The randomness rate of an $n$-bit string $x$ is $C(x)/n$.} More precisely, they show that for any positive rational numbers $\sigma$ and $\epsilon$, there exists a polynomial-time computable function $E$ and a constant $h$ such that for any $x \in \zo^n$ with $C(x) \geq \sigma n$, there exists a string $\alpha_x$ of length $h$ such that $C(E(x, \alpha_x)) \geq (1-\epsilon)m$ and $m \geq cn$, for some constant $c$ that depends on $\sigma$ and $\epsilon$. Note that this result implies that it is possible to construct in polynomial time a list with $2^h$ strings and one of them is guaranteed to have Kolmogorov complexity at least $(1-\epsilon)m$. 

The shortcoming of Fortnow et al's result is that the randomness rate of the output is not $1$. It would be desirable that $C(E(x, \alpha_x)) \geq m - o(m)$. We first remark that as a  consequence of a result of Vereshchagin and Vyugin~\cite{ver-vyu:j:kolm}, randomness rate $1$ cannot be obtained with a constant number of bits of advice about the input. Indeed, we show that if a computable function $E: \zo^n \times \zo^h \mapping \zo^m$ has the property that for all strings $x$ with $C(x) \geq \sigma n$, it holds that there exists $\alpha_x$ such that $C(E(x, \alpha_x)) \geq (1-\epsilon)m$, then $\epsilon \geq \frac{1-\sigma}{2^{h+1}-1} - o(1)$ (provided that $m = \omega(\log n +h)$).

In contrast with the above impossibility result, we show that from sources with a linear amount of randomness, one can extract a string with randomness rate $1$ with basically any non-constant amount of advice, such as, for example, the inverse of the Ackerman function. This is an instantiation of the following more general result.
\smallskip

\begin{fact}
\label{t:boundsinformal}
(Informal statement;  see Theorem~\ref{t:extractsmalladvice} for full statement.)
For any $m$ computable from $n$ there exists a computable function $E$ with the following property:

For every $n$-bit string $x$ with complexity $\geq m$, there exists a string $\alpha_x$ of length $\omega( \log \frac{n}{m})$ such that $C(E(x,\alpha_x)) = m - o(m)$ and the length of $C(E(x, \alpha_x))$ is $m$.
\end{fact}
\smallskip

Note that the function $E$ from Fact~\ref{t:boundsinformal} is computable, but no complexity bound is claimed for it.  We can obtain an extractor $E$ computable by a polynomial-size circuit with almost all the properties from 
Theorem~\ref{t:boundsinformal}. Basically the weakening is that the output length $m$ has to be $\leq cn$, for some positive constant $c$. Moreover the polynomial-size circuit is itself computable, in the sense that there exists an algorithm that on input $n$ outputs the description of the circuit that computes $E$.  We call such a circuit an \emph{effectively constructible  circuit}.
\begin{fact}
\label{t:boundspolysizeinformal}
(Informal statement;  see Theorem~\ref{t:derandextract} for full statement.)
There exists a constant $c$ such that for any $m$ computable in polynomial time from $n$ and $m \leq cn$ there exists a function $E$, which is computable by a polynomial-size effectively constructible circuit with the following property:

For every $n$-bit string $x$ with complexity $\geq m$, there exists a string $\alpha_x$ of length $\omega( \log \frac{n}{m})$ such that $C(E(x,\alpha_x)) = m - o(m)$ and the length of $E(x, \alpha_x)$ is $m$.
\end{fact}
\medskip

{\bf Discussion of technical aspects.}
We present the main ideas in the proofs of Fact~\ref{t:chain}, Fact~\ref{t:boundsinformal}, and Fact~\ref{t:boundspolysizeinformal}. As mentioned, the proofs rely on Kolmogorov extractors. A Kolmogorov extractor is a computable ensemble of functions $E: \zo^{n_1} \times \zo^{n_2} \mapping \zo^m$ such that for all $x \in \zo^{n_1}$ and $y \in \zo^{n_2}$ that have Kolmogorov complexity above a certain threshold value and that are sufficiently independent (which roughly means that $C(y \mid x) \approx C(y)$), it holds that $C(E(x,y)) \approx m$.

The general idea for the Chain Rule in Fact~\ref{t:chain} is to construct a Kolmogorov extractor $E: \zo^n \times \zo^n \mapping \zo^m$ that extracts almost all the randomness from its inputs, \ie, $C(E(x,y) \mid n) \geq C(x \mid n) + C(y \mid x) - ({\rm{small~ term}})$. If $E$ would be a computable function, then $C(E(x,y) \mid n) \leq C(xy \mid n) + O(1)$ (*), and we would obtain the difficult direction in the Chain Rule, $C(xy \mid~n) \geq C(x \mid~n) + C(y \mid~x)$, modulo the small additive term. However, in order to have that term small enough for the desired level of tightness in the Chain Rule, the randomness extraction needs to be finely tuned (this is the most technical part of the proof) depending on some attributes of $x$ and $y$.  Thus, $E$ will not be fully computable, as required by the standard definition of a Kolmogorov extractor. Instead,  $E$ needs a few bits of information about its inputs, and even if this weakens the inequality (*), it still allows us to achieve the desired level of tightness in the Chain Rule.

To obtain the extractors $E$ that require a small amount of non-uniform information about the source (see Fact~\ref{t:boundsinformal} and Fact~\ref{t:boundspolysizeinformal}), the use of Kolmogorov extractors is quite direct. The goal here is to show that for each $x$, it is enough to have a short string $\alpha_x$ such that $E(x, \alpha_x)$ has randomness rate $1$ and contains almost all the randomness of $x$. The solution is based on the fact that from $x$ and a short string that is random even conditioned by $x$, one can extract almost all the randomness of $x$. This is similar to the well-studied case of seeded extractors, with the remark that we can have shorter seeds because requiring that the output has randomness rate equal to $1$ is weaker than requiring that the output is statistically close to the uniform distribution (as stipulated in the definition of seeded extractors). Then we take $\alpha_x$ to be such a short seed. The above fact is obtained via an elementary use of the probabilistic method. We first identify a combinatorial object, called a \emph{balanced table}, that characterizes a Kolmogorov extractor, in the sense that the table of a Kolmogorov extractor must satisfy the combinatorial constraints of a balanced table. We show (with the probabilistic method) that such an object exists  with a seed of length $\omega(\log (n/m))$, where $n$ is the length of $x$ and $m$ is the Kolmogorov complexity of $x$. This establishes Fact~\ref{t:boundsinformal}. Since the function $E$ from Fact~\ref{t:boundsinformal} is obtained via the probabilistic method, we cannot claim any complexity bound for it.  To obtain the Kolmogorov extractor in Fact~\ref{t:boundspolysizeinformal}, which is computed by polynomial-size circuits,  we derandomize the construction from Fact~\ref{t:boundsinformal}, using a method of Musatov~\cite{mus:t:spacekolm}. The key observation is that the combinatorial constraints of a balanced table can be checked by constant-depth circuits of relatively small size. The argument goes as follows: (a) these constraints require that in all sufficiently large rectangles of the table no element appears too many times; (b) thus one needs to count the occurrence of each element in every sufficiently large rectangle of the table; (c) by a well-known result of Ajtai~\cite{ajt:j:constantdepthcount},  this operation can be done with sufficient accuracy by constant-depth circuits with relatively small size. Therefore, we can use the Nisan-Wigderson~(\cite{nis-wig:j:hard}) pseudo-random generator NW-gen that fools bounded-size constant-depth circuits and has seeds of size polylogaritmic in the size of the output.  Since balanced tables with the required parameters are abundant, we infer that there exists a seed $s$ so that NW-gen$(s)$ is a balanced table with the required parameters. A balanced table is an object of size exponential in $n$, which implies that the seed $s$ has size polynomial in $n$. Moreover, the Nisan-Wigderson pseudo-random generator has the property that each bit of the output can be calculated separately in time polynomial in the length of the seed. This implies that the Kolmogorov extractor whose table is NW-gen$(s)$ can be computed by a polynomial-sized circuit that has $s$ hard-wired in its circuitry.

\section{Preliminaries}
\subsection{Notation and basic facts on Kolmogorov complexity}
The Kolmogorov complexity of a string $x$ is the length of the shortest effective description of $x$. There are several versions of this notion. We use here  the \emph{plain complexity}, denoted $C(x)$, and also the \emph{conditional plain complexity} of a string $x$ given a string $y$, denoted $C(x \mid y)$, which is the length of the shortest effective description of $x$ given $y$. The formal definitions are as follows.
We work over the binary alphabet $\zo$. A string is an element of $\{0,1\}^*$.
If $x$ is a string, $|x|$ denotes its length.  
Let $M$ be a Turing machine that takes two input strings and outputs one string. For any strings $x$ and $y$, define the \emph{Kolmogorov complexity} of $x$ conditioned by $y$ with respect to $M$, as 
$C_M(x \mid y) = \min \{ |p| \mid M(p,y) = x \}$.
There is a universal Turing machine $U$ with the following property: For every machine $M$ there is a constant $c_M$ such that for all $x$, $C_U(x \mid y) \leq C_M(x \mid y) + c_M$.
We fix such a universal machine $U$ and dropping the subscript, we write $C(x \mid y)$ instead of $C_U(x \mid y)$. We also write $C(x)$ instead of $C(x \mid \lambda)$ (where $\lambda$ is the empty string). The \emph{randomness rate} of a string $x$ is defined as ${\rm rate}(x) = \frac{C(x)}{|x|}$.  If $n$ is a natural number, $C(n)$ denotes the Kolmogorov complexity of the binary representation of $n$. For two $n$-bit strings $x$ and $y$, the information in $x$ about $y$ is denoted $I(x : y)$ and is defined as $I(x : y) = C(y \mid n) - C(y \mid x)$.

In this paper, the constant hidden in the $O(\cdot)$ notation only depends on the universal Turing machine.

For all $n$ and $k \leq n$, 

$2^{k-O(1)} < |\{x \in \zo^n \mid C(x\mid~n) < k\}| < 2^k$.


Strings $x_1, x_2, \ldots, x_k$ can be encoded in a self-delimiting way (\ie, an encoding from which each string can be retrieved) using $|x_1| + |x_2| + \ldots + |x_k| + 2 \log |x_2| + \ldots + 2 \log |x_k| + O(k)$ bits. For example, $x_1$ and $x_2$ can be encoded as $\overline{(bin (|x_2|)} 01 x_1 x_2$, where $bin(n)$ is the binary encoding of the natural number $n$ and, for a string $u = u_1 \ldots u_m$, $\overline{u}$ is the string $u_1 u_1 \ldots u_m u_m$ (\ie, the string $u$ with its bits doubled).

\if01
The Symmetry of Information Theorem (see \cite{zvo-lev:j:kol}) states that for all strings $x$ and $y$, $C(xy) \approx C(y) + C(y \mid x)$. More precisely:
$| (C(xy) - (C(x) + C(y \mid x)) | \leq O( \log C(x) + \log C(y))$.
In case the strings $x$ and $y$ have length $n$, it can be shown that
$| (C(xy) - (C(x) + C(y \mid x)) | \leq 2 \log n + O(\log \log n)$.
\fi

All the Kolmogorov extractors in this paper are ensembles of functions $f = (f_n)_{n \in \nat}$ of type $f_n : \zo^n \times \zo^{k(n)} \mapping \zo^{m(n)}$. For readability, we usually drop the subscript and the expression ``ensemble $f: \zo^n \times \zo^k \mapping \zo^m$'' is a substitute for
``ensemble $f = (f_n)_{n\in \nat}$, where for every $n$, $f_n : \zo^n \times \zo^{k(n)} \mapping \zo^{m(n)}$.''

For any $n \in \nat$, $[n]$ denotes the set $\{1,2, \ldots, n\}$.

\subsection{Approximate counting via polynomial-size constant-depth circuits}
In the derandomization argument used in the proof of Theorem~\ref{t:derandextract}, we need to count with constant-depth polynomial-size circuits. Ajtai~\cite{ajt:j:constantdepthcount} has shown that this can be done with sufficient precision.

\begin{theorem}
\label{t:ajtai}
(Ajtai's approximate counting with polynomial size constant-depth circuits.)
There exists a uniform family of circuits $\{G_n\}_{n \in \nat}$, of polynomial size and constant depth, such that for every $n$, for every $x \in \zo^n$, for every $a \in \{0, \ldots, n-1\}$, and for every $\epsilon > 0$,
\begin{itemize}
	\item If the number of $1$'s in $x$ is $\leq (1 - \epsilon)a$, then $G_n(x,a,1/\epsilon) = 1$,
	\item If the number of $1$'s in $x$ is $\geq (1 + \epsilon)a$, then $G_n(x,a,1/\epsilon) = 0$.
\end{itemize}
\end{theorem}
 We do not need the full strength (namely, the uniformity of $G_n$) of this theorem; the required level of accuracy (just $\epsilon > 0$) can be achieved by non-uniform polynomial-size circuits of depth $d=3$ (with a much easier proof, see~\cite{vio:t:approxcount}).

\subsection{Pseudo-random generator fooling bounded-size constant-depth circuits}
The derandomization in the proof of Theorem~\ref{t:derandextract} is done using the Nisan-Wigderson pseudo-random generator that ``fools" constant-depth circuits~\cite{nis-wig:j:hard}. Typically, it is required that the circuit to be fooled has polynomial size, but the proof works for  circuits of size $2^{n^\alpha}$ for some small constant $\alpha > 0$.
\begin{theorem}
\label{t:NWgen}
(Nisan-Wigderson pseudo random generator.)
For every constant $d$ there exists a constant $\alpha > 0$ with the following property. There exists a function $\mbox{NW-gen}:\zo^{O(\log^{2d+6}n)} \mapping \zo^n$ such that for any circuit $G$ of size $2^{n^\alpha}$ and depth $d$,
\[
 | \prob_{s \in \zo^{O(\log^{2d+6}n)}}[G(\mbox{NW-gen}(s)) = 1] - \prob_{z \in \zo^n} [G(z)=1]| < 1/100.
\]
Moreover, there is a procedure that on inputs $(n, i,s)$ produces the $i$-th bit of $\mbox{NW-gen}(s)$ in time
$\poly(\log n)$.

\end{theorem}
\section{Symmetry of information for random strings}
 \label{s:syminf}
 We start by proving the new form of the Chain Rule in Fact~\ref{t:chain}. The formal statement is as follows:
 \smallskip
 
 \begin{theorem}
 \label{t:lowerboundsym}
 For all strings $x \in \zo^n$ and $y \in \zo^n$ with $C(x \mid n) \geq 13 \log n +  I(x:y) + O(1)$ and $C(y \mid n) \geq 7 \log n + I(x:y) + O(1)$,
 \[
 \begin{array}{ll}
 C(xy \mid n) & \geq C(x \mid n) + C(y \mid x) - \log I(x : y)   \\
        & \quad \quad - O(C^{(2)}(x \mid n) + C^{(2)}(y \mid n) + C^{(2)}(y \mid x) ). 
 \end{array}
 \]
 \end{theorem}
 \smallskip
 
 \begin{proof} Let $x$ and $y$ be as in the statement of the theorem. We introduce some notation:
 
\begin{itemize}
	\item $t_x = C(x \mid n)$,
	\item $t_y = C(y \mid n)$,
	\item $t_{y,x} = C(y \mid x)$,
	\item $d = 2(C(t_x \mid n) + C(t_y \mid n) + C(t_{y,x} \mid n)) + I(x :~y) + O(1)$, 
	
	(note that $d= O(C^{(2)}(x \mid n) + C^{(2)}(y \mid n) + C^{(2)}(y \mid~x)) + I(x : y)$),
	\item $k_x = t_x - 2(C(t_x \mid n) + C(t_y \mid n) + C(t_{y,x} \mid n)) - O(1)$,
	\item $k_y = t_y - O(1)$,
	\item $m = k_x + k_y - \log d - O(1)$,
\end{itemize}
 where the constant $O(1)$ only depends on the universal machine and will be chosen later.
 We also denote $K_x = 2^{k_x}$, $K_y = 2^{k_y}$, $N = 2^n$, $D= 2^d$ and $M =2^m$. We  state two claims, which we prove later, that immediately establish the theorem. The first claim shows the existence of a function $E : \zo^n \times \zo^n \mapping \zo^m$ that satisfies certain combinatorial constraints similar to those of a balanced table, whose parameters are tailored for what we need. We view $E$ as an $[N]$-by-$[N]$ table colored with colors in $[M]$. A rectangle $B_1 \times B_2$, where $B_1 \subseteq [N]$ and $B_2 \subseteq [N]$, is the part of the table formed by the rows in $B_1$ and the columns in $B_2$. For $A \subseteq [M]$, we say that a cell $(u,v)$ of the table is an $A$-cell, if $E(u,v) \in A$. 
 
 \begin{claim}
 \label{c:balancedt}
 There exists $E: \zo^n \times \zo^n \mapping \zo^m$ such that for every rectangle $B_1 \times B_2$, with $|B_1| \geq K_x$, $|B_2| \geq K_y$, and for every $A \subseteq \zo^m$ with
 $|A| \geq \frac{M}{D}$, the number of $A$-cells in $B_1 \times B_2$ is $\leq 2 \cdot \frac{|A|}{M} \cdot |B_1| \cdot |B_2|$.
 \end{claim}

 Note that given $n, t_x, t_y, t_{y,x}$, $C(t_x \mid~n), C(t_y \mid~n)$, $C(t_{y,x} \mid~n)$, one can effectively enumerate the tables that satisfy the claim. Let $E$ be the first table that appears in this enumeration. We denote the tuple $(t_x, t_y, t_{y,x}, C(t_x \mid~n), 
 C(t_y \mid~n), C(t_{y,x} \mid~n))$ by $\Lambda$. Given $n$, the tuple $\Lambda$ can be encoded in a self-delimiting way using $C(t_x \mid~n) + C(t_y \mid~n)+ C( t_{y,x}\mid~n) + C(C(t_x \mid~n) + C(C(t_y \mid~n)) + C( C(t_{y,x} \mid~n)) + 2 \log C(t_x \mid~n) + 2 \log C(t_y \mid~n)+ 2 \log C( t_{y,x} \mid~n) + 2 \log C(C(t_x \mid~n) + 2 \log C(C(t_y \mid~n)) + 2 \log C( C(t_{y,x} \mid~n)) + O(1)$ bits, a value which we denote by $\lambda$. Using $C(C(t_x \mid n)) \leq \log C(t_x \mid n)$ and the other similar inequalities, we see that for an appropriate choice of the constants, $d > \lambda + I(x : y)$. 
 
 The second claim shows that $E$ extracts almost all the randomness in $x$ and $y$.
 \begin{claim}
 \label{c:extract}
 $C(E(x,y) \mid n,m) \geq m-d$.
 \end{claim}
 The theorem follows easily from the two claims. Note that the calculation of $E(x,y)$ requires a description of $x$, $y$ and $\lambda$ bits for the self-delimited description of $E$. Then
 \[
 \begin{array}{ll}
 C(xy \mid n) & \geq C(E(x,y) \mid n, m) - \lambda - O(1)\\
  & \geq m - d - \lambda - O(1)\\
 & = k_x + k_y - \log d  - d - \lambda - O(1) \\
 & = t_x + t_y - I(x : y) \\
 & \quad\quad  - 6 (C^{(2)}(x \mid n) + C^{(2)}(y \mid n) + C^{(2)}(y \mid x) ) \\
 &  \quad\quad  - \log ( 2 (C^{(2)}(x \mid n) + C^{(2)}(y \mid n) + C^{(2)}(y \mid x)) + I(x : y)) - O(1) \\
 & \geq t_x + t_y - I(x : y) - \log I(x : y) \\
 & \quad\quad  - O (C^{(2)}(x \mid n) + C^{(2)}(y \mid n) + C^{(2)}(y \mid x) ), 
 \end{array}
 \]
where in the last line we have used the fact that 
\[
\begin{array}{ll}
\log ( 2 (C^{(2)}(x \mid n) + C^{(2)}(y \mid n) + C^{(2)}(y \mid x)) + I(x : y)) & \\
\leq \log ( 2 (C^{(2)}(x \mid n) + C^{(2)}(y \mid n) + C^{(2)}(y \mid x))) + \log I(x : y).& \\
\end{array}
\]
 Since $t_x = C(x \mid n)$ and $t_y - I(x : y) = C(y \mid~n) - (C(y \mid~n) - C(y \mid~x))= C(y \mid~x)$, the conclusion follows.
\smallskip

 It remains to prove the two claims.
 \smallskip
 
 \emph{Proof of Claim~\ref{c:balancedt}.} We use the probabilistic method. The details are presented in the Appendix.
 \smallskip
 
 \emph{Proof of Claim~\ref{c:extract}.} Suppose that $C(E(x,y) \mid n, m) < m-d$.
 Let $A = \{w \in \zo^m \mid C(w \mid n, m) \leq m-d + O(1)\}$, where the constant $O(1)$ (depending only on the universal machine) is chosen so that $|A| \geq 2^{m-d} = \frac{M}{D}$. It also holds that $|A| \leq 2^{m-d + O(1)}$.
 
We define $B_{y} = \{v \in \zo^n \mid C(v \mid n) \leq t_y\}$.
 For a convenient choice of the constant appearing in the definition of $k_y$, it holds that $|B_{y}| \geq K_y$. Also note that $|B_{y}| \leq 2^{t_y+1}$.
 
 We say that a row $u \in \zo^n$ is \emph{bad}, if the number of $A$-cells in the $\{u\} \times B_{y}$ rectangle of $E$ is $ > 2 \cdot \frac{|A|}{M} \cdot |B_{y}|$.
 
 The number of bad rows is bounded by $K_x$ (otherwise, $E$ would not satisfy the requirement in Claim~\ref{c:balancedt} for the rectangle formed by the bad rows and $B_y$). Note that, given $n, t_x, t_y, t_{y,x}, C(t_x \mid n), C(t_y \mid n), C(t_{y,x} \mid n)$, one can enumerate the set of bad rows. Therefore a bad row $u$ can be described by its rank in an enumeration of the set of bad rows, and by the information $\Lambda$ required to run this enumeration. Therefore if $u$ is a bad row, then
 \[
 \begin{array}{ll}
 C(u \mid n) & \leq k_x + \lambda \\
 & <  t_x.
 \end{array}
 \]
 Since $C(x \mid n) = t_x$, $x$ is a good row. Therefore the number of $A$-cells in $\{x \} \times B_{y}$ is
 \[
 \begin{array}{ll}
 < 2 \cdot \frac{|A|}{M} \cdot |B_{y}| \\
 < 2 \cdot \frac{2^{m-d + O(1)}}{2^m} \cdot 2^{t_y+1} \\
 = 2^{t_y - d +O(1)}.
 \end{array}
\]
 By our assumption, $(x,y)$ is an $A$-cell, and, obviously, it is in the $\{x\} \times B_y$ rectangle. Given $x$, $y$ can be described by the rank of $(x,y)$ in an enumeration of $A$-cells in $\{x\} \times B_{y}$ and by the information $\Lambda$ required to run this enumeration.
 
 Thus,
 \[
 \begin{array}{ll}
C(y \mid x)  &\leq t_y - d + O(1) + \lambda \\
 & < t_y - I(x : y) \\
 & = C(y \mid n) - (C(y \mid n) - C(y \mid x)) \\
 & = C(y \mid  x),
 \end{array}
 \]
 which is a contradiction.
 \qed \end{proof}
 \smallskip
 
From Theorem~\ref{t:lowerboundsym}, it is easy to derive the strong Symmetry of Information relation for random strings stated in Fact~\ref{t:symrandom}. The formal statement is as follows.
\begin{theorem}
\label{t:symrandstrings}
For every constant $c \geq 0$, there exists a constant $c' \geq 0$ with the following property: For every $n \in \nat$, for every $c$-random string $x \in \zo^n$ and every $c$-random string $y \in \zo^n$, if $y$ is $c$-random conditioned by $x$, then $x$ is $c'$-random conditioned by $y$.

\end{theorem} 
 \begin{proof}  Note that if $C(x \mid n) \geq n-c$ and $C(y \mid x) \geq n-c$, then $C^{(2)}(x \mid n) = O(1)$, $C^{(2)}(y \mid n) = O(1)$ and $C^{(2)}(x \mid n) = O(1)$. Then $I(x : y) = O(1)$.
 
 Therefore, from Theorem~\ref{t:lowerboundsym}, we obtain $C(xy \mid n) \geq 2n - O(1)$. We also have
 $C(xy \mid~n) \leq C(y \mid~n) + C(x \mid~y) + 2 C (C(y \mid~n)) + O(1) = C(y \mid~n) + C(x \mid~y) +O(1)$. Thus, $C(x \mid~y) \geq 2n - C(y \mid~n) - O(1) = n - O(1)$.
 \qed \end{proof}

\section{Randomness extraction with small advice}
In this section we study the amount of non-uniformity that is necessary for randomness extraction from a single source of randomness.

Vereshchagin and Vyugin~\cite{ver-vyu:j:kolm} show the limitations of what can be extracted with a bounded quantity of advice. To state their result, let us fix $n = $ length of the source, $h =$ number of bits of advice that is allowed, and $m = $ the number of extracted bits. Let  $H = 2^{h+1}-1$.
\smallskip

\begin{theorem}[\cite{ver-vyu:j:kolm}]
\label{t:vervyu}
There exists a string $x \in \zo^n$  with $C(x) > n- H \log (2^m+1) \approx n - Hm$ such that any string $z \in \zo^m$ with $C(z \mid x) \leq h$ has complexity $C(z) < h + \log n + \log m + O(\log \log n, \log \log m)$. 
\end{theorem}
\smallskip

The next theorem, a consequence of Theorem~\ref{t:vervyu}, shows that no Kolmogorov extractor for sources with randomness rate $\sigma$ and that uses $h$ bits of advice about the source can output strings with randomness rate larger than $1 - (1-~\sigma)/H$.
\smallskip

\begin{theorem}
\label{t:limits}
Assume that the parameters $m, h, \sigma$ are computable from $n$ and satisfy the following relations: $0 < \sigma < 1, h > 0, 0 < m < n$, $m = \omega (\log n + h)$. 

Let $f: \zo^n \times \zo^h \mapping \zo^m$ be a computable ensemble of functions such that for every $ x \in \zo^n$ with $C(x) \geq \sigma \cdot n$, there exists a string $\alpha_x$ such that $C(f(x , \alpha_x)) \geq (1-\epsilon)\cdot m$. Then $\epsilon \geq \frac{1-\sigma}{H}- o(1)$.
\end{theorem}
\smallskip

\begin{proof} Let $m' = \min (\lfloor \frac{1-\sigma}{H} \cdot n\rfloor, m)$. Note that $m' \geq  \lfloor \frac{1-\sigma}{H} \cdot m\rfloor$.
Let $x$ be the string guaranteed by the Vereshchagin-Vyugin Theorem~\ref{t:vervyu} for the parameters $n, h + c, m'$, where $c$ is a constant that will be specified later. Note that $C(x) > n - H\cdot m' \geq \sigma \cdot n$. By assumption there is a string $\alpha_x$ such that $C(f(x, \alpha_x)) \geq (1-\epsilon)m$. Let $z$ be the prefix of length $m'$ of $f(x,\alpha_x)$. Note that $C(f(x, \alpha_x)) \leq C(z) + (m-m') + 2\log m + O(1)$, which implies that
$C(z) \geq (1-\epsilon)m - m + m' - 2 \log m - O(1) \geq \frac{(1-\sigma)m}{H} - \epsilon m - 2 \log m - O(1)$.

We also have $C(z \mid x) \leq |\alpha_x| + c = h + c$, for some constant $c$. It follows from Theorem~\ref{t:vervyu} that
$C(z) < h + \log n + \log m' + O(\log \log n, \log \log m')$. So,
$\frac{(1-\sigma)}{H}m - \epsilon m - 2 \log m - O(1) \leq h + \log n + O(\log \log n, \log \log m')$, which implies that $\epsilon \geq \frac{1-\sigma}{H} - \frac{h + O(\log n)}{m} = \frac{1-\sigma}{H} - o(1)$.
\qed \end{proof}
\smallskip

We move to showing the positive results in Fact~\ref{t:boundsinformal} and Fact~\ref{t:boundspolysizeinformal} regarding randomness extraction with small advice that complement the negative result in Theorem~\ref{t:limits}. The constructions use the parameters $n, n_1, m, k, \delta$ and $d$. We denote $N=2^n, N_1 = 2^{n_1}, M=2^m, \Delta = 2^{\delta}$ and $D= 2^{d}$. We identify in the natural way a function $E: \zo^n \times \zo^{n_1} \mapping \zo^m$ with an $[N] \times [N_1]$ table colored with colors from $[M]$. For $A \subseteq [M]$, we say that an $(u,v)$ cell of the table is an $A$-cell if $E(u,v) \in A$. The reader might find helpful to consult the proof plan presented in the Introduction. As explained there the notion of a \emph{balanced table} plays an important role.
\begin{definition}
A table $E: [N] \times [N_1] \mapping [M]$ is $(K,D,\Delta)$-balanced if for any $B \subseteq [N]$ with $|B| \geq K$, for any $A \subseteq [M]$ with $\frac{|A|}{M} \geq \frac{1}{D}$, it holds that
\[
\frac{|\mbox{$A$-cells in $B \times [N_1]$}|}{|B| \times N_1} \leq \Delta \cdot \frac{|A|}{M}.
\]
\end{definition}
The following lemma shows that a balanced table is a good Kolmogorov extractor.
 \begin{lemma}
 \label{l:tableext}
  Let $E: [N] \times [N_1] \mapping [M]$
be a $(K,D,\Delta)$-balanced table and $d = \delta + O(1)$. Suppose $C(E \mid n) = O(1)$ and $n_1, k, d$, and $\delta$ are computable from $n$. Let $(x,y) \in [N]\times [N_1]$ be such that $C(x \mid n) \geq k + O(1)$ and $C(y \mid x ) \geq n_1$. Let $z = E(x,y)$. Then $C(z \mid m) > m-d$. 

(Note: $O(1)$ means that there exist constants, depending only on the universal machine, for which the statements hold.)
  \end{lemma}
 \begin{proof} The proof is similar to the proof of Claim~\ref{c:extract}. We sketch the argument. Suppose $C(z \mid m) \leq m-d$. 
 
 Let $A = \{w \in \zo^m \mid C(w \mid m) \leq m-d + O(1)\}$, where the constant $O(1)$ is chosen so that $|A| \geq 2^{m-d}$. Also note that $|A| \leq 2^{m-d+O(1)}$. 
 
 We say that a row $v$ is \emph{bad} if the number $A$-cells in the $\{v \} \times [N_1]$ rectangle of $E$ is $> \Delta \cdot \frac{|A|}{M} \cdot N_1$.
 The number of bad rows is at most $K$, because the table $E$ is $(K,D,\Delta)$-balanced.
 Therefore a bad row $v$ is described by the information needed to enumerate the bad rows (and this information is derivable from $n$) and from its rank in the enumeration of bad rows. So, if $v$ is bad, $C(v \mid n) < k+O(1)$. 
 
 Since $C(x \mid n) > k + O(1)$, it follows that $x$ is good. Therefore, the number of $A$-cells in the $\{x\} \times [N_1]$ rectangle of $E$ is $\leq \Delta \cdot \frac{|A|}{M} \cdot N_1 =2^{\delta - d + n_1 + O(1)}$.
 
 Note that, by our assumption, the cell $(x,y)$ is an $A$-cell. Then the string $y$, given $x$, can be described by the rank of $(x,y)$ among the $A$-cells in the $\{ x \} \times [N_1]$ rectangle of $E$.
 
 So, $C(y \mid x) \leq \delta - d + n_1 +O(1)$ and the right hand side is less than $n_1$ for an appropriate choice of the constant $O(1)$ in the relation between $d$ and $\delta$. We obtain that $C(y \mid x) < n_1$, contradiction.
 \qed \end{proof}
 \smallskip
 
 The next lemma establishes the parameters for which balanced tables exist.
 \begin{lemma}
 \label{l:probtable}
 Suppose the parameters satisfy the following relations: $D = O(\Delta), n/\delta = o(N_1)$, and $M = o(\delta\cdot K \cdot N_1)$. Then there exists a table $E:[N] \times [N_1] \mapping [M]$ that is $(K,D,\Delta)$-balanced.
 \end{lemma}
 \begin{proof}  The proof is by the probabilistic method and is presented in the Appendix.
 \qed \end{proof}
 \smallskip
 
 We can now prove Fact~\ref{t:boundsinformal}. The formal statement is as follows.
 \begin{theorem}
 \label{t:extractsmalladvice}
 Parameters: Let $m(n)$ and $h(n)$ be computable functions such that $m(n) < n$ for all $n$ and $h(n) = \omega (\log \frac{n}{m(n)})$.

 There exists a computable function $E:\zo^n \times \zo^{h(n)} \mapping \zo^{m(n)}$, such that for every $x \in \zo^n$ with $C(x \mid n) \geq m(n)$, there exists $\alpha_x \in \zo^{h(n)}$ such that
 $C(E(x, \alpha_x) | m) \geq m(n) - o(m(n))$.
 
 \end{theorem}
 \begin{proof} We take $\delta = \frac{n}{2^{0.5 h(n)}}$, $d = \delta + c$, where $c$ is the constant from Lemma~\ref{l:tableext}, $n_1 = h(n)$.
 
 By Lemma~\ref{l:probtable}, there exists a table $E: [N] \times [N_1] \mapping [M]$ that is
 $(K,D,\Delta)$-balanced, and by brute force one can build such a table from $n$. Thus we obtain such a table $E$ with $C(E \mid n) = O(1)$. We take $\alpha_x$ to be a string in $\zo^{h(n)}$ such that $C(\alpha_x \mid x) \geq h(n)$. Using Lemma~\ref{l:tableext}, we obtain that $C(E(x, \alpha_x) \mid m) \geq m(n) - d = m(n) -  \frac{n}{2^{0.5 h(n)}} - c =  m- o(m)$.
 \qed \end{proof}
 \smallskip
 
 Our next goal is to derandomize the construction in Theorem~\ref{t:extractsmalladvice}. As explained in the Introduction the key observation is that checking if a table is balanced can be done, in an approximate sense, by constant-depth circuits with relatively small size.
 \smallskip
 
 \begin{lemma}
 \label{l:countcircuit}
 The parameters $n_1,m,k, d$, and $\delta$ are positive integers computable from $n$  in polynomial time.
 We assume $k \leq n, m \leq k, d \leq n$.
 
 There exists a circuit $G$ of size $\poly(N^K)$ and constant depth such that for any table $E: [N] \times [N_1] \mapping [M]$,
 
 (a) if $G(E) = 1$, then $E$ is $(K,D,1.03 \Delta)$-balanced,
 
 (b) if $E$ is $(K,D, \Delta)$-balanced, then $G(E) = 1$.
 
 \end{lemma}
 \smallskip
 
 \begin{proof} Let $a = (1/0.99) \Delta \cdot 1/D \cdot K \cdot N_1$. Let us fix for the moment  a set
 of rows $B \subseteq [N]$ of size $|B| = K$ and a set of colors $A \subseteq [M]$ of size $|A| = M/D$. Let $x_{B,A}$ be a binary string indicating which cells in the $B \times [N_1]$ rectangle of $E$ are $A$-colored. Formally, $x_{B,A}$ is the string of length $K \cdot N_1$, whose $\langle i, j \rangle$-th bit is $1$ if the cell $(i,j)$ in the rectangle $B \times [N_1]$ of $E$ is an $A$-cell and $0$ if it is not.

 By Ajtai's Theorem~\ref{t:ajtai}, there exists a polynomial-size constant-depth circuit $G'$ (which does not depend on $B$ and $A$) with $a$ hardwired and such that 
 \begin{itemize}
	\item  $G'(x_{B,A}) = 1$ if the number of $A$-cells in $B \times [N_1]$ is at most $(1-0.01) \cdot a$, and
	\item $G'(x_{B,A}) = 0$ if the number of $A$-cells in $B \times [N_1]$ is at least $(1+0.01) \cdot a$.
\end{itemize}
Now we describe the circuit $G$.

 The circuit $G$ on input an encoding of the table $E$ (having length $N \cdot N_1 \cdot m$) computes in constant depth a string $x_{B,A}$  for every $B \subseteq [N]$ with $|B| = K$ and for every $A \subseteq [M]$ with $A = M/D$. There are ${N \choose K} {M \choose M/D} = \poly(N^K)$ such strings $X_{B,A}$. Each such string $x_{B,A}$ is the input of a copy of $G'$. The output gates of all the copies of $G'$ are connected to an AND gate, which is the output gate.
 
 If $G(E) = 1$, then $G'(x_{B,A}) = 1$ for all $B$'s and $A$'s as above. This implies that for all $B \subseteq [N]$ with $|B| \geq K$ and all $A \subseteq [M]$ of size $\geq M/D$, the number of $A$-cells in the $B \times [N_1]$ rectangle of $E$ is at most $(1+0.01)a \leq (1.03) \cdot \Delta \cdot (1/D) \cdot K \cdot N_1$, \ie, $E$ is $(K, D, 1.03\Delta)$-balanced.
 
 In the other direction, if $E$ is $(K,D, \Delta)$-balanced then for all $B \subseteq [N]$ with $|B| = K$ and for all $A \subseteq [M]$ with $A = M/D$, the number of $A$-cells in $B \times [N_1]$ is at most $\Delta \cdot (1/D) \cdot K \cdot N_1 = (1-0.01)a$, which implies that $G(E) =1$.
\qed \end{proof}
 \smallskip
 
 We next prove Fact~\ref{t:boundspolysizeinformal}. The formal statement is as follows.
\smallskip

 \begin{theorem}
 \label{t:derandextract}
 Parameters: Let $m(n)$ and $h(n)$ be polynomial-computable functions such that $m(n) \leq 0.99 \cdot \alpha_{NW} \cdot  n$ and
 $h(n) = \omega (\log(\frac{n}{m(n)}))$. 
 
 There exists a function $E: \zo^n \times \zo^{h(n)} \mapping \zo^{m(n)}$,  computable by an  effectively constructible circuit having polynomial size and the following property: For every $x \in \zo^n$ with $C(x \mid n) \geq m(n) + O(1)$, there exists a string $\alpha_x \in \zo^{h(n)}$ such that $C( E(x, \alpha_x) \mid m) \geq m - o(m)$.
 \end{theorem}
 \smallskip

 \begin{proof} Let $k = m(n), \delta = \frac{n}{2^{0.5 h(n)}}, d = \delta + c + \log 1.03$  (where $c$ is the constant from Lemma~\ref{l:tableext}), and  $n_1 = h(n)$.
 
  Let $G$ be the circuit promised by Lemma~\ref{l:countcircuit} for these parameters.  Let $d_{Ajtai}$ be the depth of the circuit $G$ and let $\alpha_{NW}$ be the constant corresponding to $d_{Ajtai}$ in Theorem~\ref{t:NWgen}.
  
  Let $\tilde{N} = N \cdot N_1 \cdot m$. This is the size of an encoding of a table $E: [N] \times [N_1] \mapping [M]$. Let $\mbox{NW-gen} : \zo^{\log^{2d _{Ajtai+6}}(\tilde{N})} \mapping \zo^{\tilde{N}}$ be the Nisan-Wigderson pseudo-random generator given by Theorem~\ref{t:NWgen} for the depth parameter equal to $d_{Ajtai}$.
  Note that $\poly(N^K) \leq 2^{\tilde{N}^{\alpha_{NW}}}$, where $\poly(N^K)$ is the bound from Lemma~\ref{l:countcircuit} for the size of the circuit $G$.

 The probabilistic argument in Lemma~\ref{l:probtable} can be modified to show that among the 
 tables of type $E: [N] \times [N_1] \mapping [M]$ the fraction of those which are $(K,D, \Delta)$-balanced is at least $0.51$.  Since $G$ accepts all such tables,
 \[
  \prob_{E \in \zo^{\tilde{N}}}[ G(E) = 1] \geq 0.51.
  \]
  Since the circuit $G$ has depth equal to $d_{Ajtai}$ and size bounded by $2^{\tilde{N}^{\alpha_{NW}}}$, it follows that if we replace a random $E \in \zo^{\tilde{N}}$ by $\mbox{NW-gen}(s)$ for a random seed $s \in \zo^{\log^{2d_{Ajtai}+6}(\tilde{N})}$, we obtain
 \[
\prob_{s \in \zo^{\log^{2d_{Ajtai}+6}(\tilde{N})}}[ G(\mbox{NW-gen}(s)) = 1] \geq 0.5.
 \]
 We only need the fact that there exists a string $s \in \zo^{\log^{2d_{Ajtai}+6}(\tilde{N})}$
 such that $\mbox{NW-gen}(s)$ is a table $E: [N] \times [N_1] \mapping [M]$ that is $(K,D, 1.03 \Delta)$-balanced. We fix such an $s$ that is computable from $n$ (say, the smallest $s$ that has the property) and the corresponding table $E$ produced by the Nisan-Wigderson pseudo-random generator on seed $s$.
 
Let us consider $x \in \zo^n$ with $C(x \mid n) \geq k$ and $\alpha_x \in \zo^{n_1}$ with $C(\alpha_x \mid y)\geq n_1$. Since $E$ is $(K, D, 1.03\Delta)$-balanced, it follows from Lemma~\ref{l:tableext} that $C(E(x,\alpha_x)\mid m) \geq m - d = m - o(m)$.

Now, let us view $E$ (which is $\mbox{NW-gen}(s)$) as a function ${\rm E}: \zo^n \times \zo^{n_1} \mapping \zo^m$. From the properties (\ie, the ``Moreover ..." in Theorem~\ref{t:NWgen}) of the Nisan-Wigderson pseudo-random generator, it follows that this function can be computed by a polynomial-size circuit which has $s$ hardwired. Since $s$ is also computable from $n$,  one can compute a description of the circuit, \ie, the circuit is effectively constructible. 
\qed \end{proof}

 \bibliography{c:/book-text/theory}

\bibliographystyle{alpha}

\newpage

\appendix
 
 \section{Appendix}
\medskip

 {\bf Proof of Claim~\ref{c:balancedt}.} 
 \smallskip
 
 We use the probabilistic method. It is enough to show the assertion for all $B_1$, $B_2$ and $A$ having sizes exactly $K_x$, $K_y$, and respectively $\frac{M}{D}$. Let us consider a random function $E: \zo^n \times \zo^n \mapping \zo^m$. Fix $B_1, B_2$ and $A$, satisfying the above requirement on their sizes. By the Chernoff's bound,
 \[
 \begin{array}{ll}
 \prob[|\mbox{$A$-cells in $B_1 \times B_2$}| \geq 2 \cdot \frac{|A|}{M} \cdot |B_1| \cdot |B_2| ] & \\
 \quad\quad\quad\quad \leq e^{-(1/3) (|A|/M)|B_1||B_2|} = e^{-(1/3)(1/D)K_x K_y}. &
 \end{array}
 \]
 The sets $B_1$, $B_2$, and $A$ can be chosen in ${N \choose K_x} \cdot {N \choose K_y} \cdot {M \choose M/D} \leq N^{K_x} \cdot N^{K_y} \cdot (eD)^{M/D} = e^{K_x \ln N + K_y \ln N + (M/D) + (M/D) \ln D}$ ways.
 Since $t_x \geq 13 \log n + I(x : y) + O(1)$  and $d < 6 \log n + I(x : y) + O(1)$, it follows that $t_x \geq d + 7 \log n +O(1)$ and from here $k_x \geq d + \log n + O(1)$. Since $t_y \geq 7 \log n + I(x : y) + O(1)$, it follows that $k_y \geq d + \log n  + O(1)$. It can be easily checked that, given these bounds for $k_x$ and $k_y$ and for an appropriate choice of the constant in the definition of $d$,
 \[
 e^{-(1/3)(1/D)K_x K_y} \cdot e^{K_x \ln N + K_y \ln N + (M/D) + (M/D) \ln D} < 1.
 \]
 Thus, the probability that a random $E$ satisfies the requirements is less than $1$, which implies that there exists an $E$ satisfying the claim.~\qed
 \medskip
 
 {\bf Proof of Lemma~\ref{l:probtable}   .}
 \smallskip
 
  The proof is by the probabilistic method. Consider a random function $E:[N]\times [N_1] \mapping [M]$.  We evaluate the probability that $E$ fails to be $(K, D, \Delta)$-balanced. Note that if $E$ fails to be $(K,D,\Delta)$-balanced, then there exists a set $B \subseteq [N]$ of size exactly $K$ and a set $A \subseteq [M]$ of size exactly $M/D$ such that the fraction of $A$ cells in the $B \times [N_1]$ rectangle of $E$ is greater  than $\Delta \cdot |A|/M$. Let us call this latter event ${\cal S}$. We show that the probability of  ${\cal S}$ is less than $1$.  Fix 
 $B \subseteq [N]$ of size $K$ and $A \subseteq [M]$ of size $M/D$. For a fixed $(x,y) \in B \times [N_1]$, $\prob[E(x,y) \in A]=|A|/M$. The expected number of $A$-cells in $B \times [N_1]$ is $\mu = |B| \cdot N_1 \cdot |A|/M$. Let $\Delta'=\Delta -1$. 
 
 We use the following version of the Chernoff bound. If $X$ is a sum of independent Bernoulli random variables, and the expected value $E[X]= \mu$, then
$\prob[X \geq (1+\Delta)\mu] \leq e^{-\Delta (\ln (\Delta/3)) \mu}$.\footnote{The standard Chernoff inequality $\prob(X \geq (1+\Delta) \mu] \leq \big( \frac{e^\Delta}{(1+\Delta)^{(1+\Delta)}}\big)^\mu$ is presented in many textbooks. It can be checked easily that $\frac{e^\Delta}{(1+\Delta)^{(1+\Delta)}} < e^{-\Delta \ln (\Delta/3)}$.}
 
 Using these Chernoff bounds, 
 \[
 \prob [|\mbox{$A$-cells in $B \times [N_1]|$} > (1+\Delta')\mu]\leq e^{-\Delta' (\ln (\Delta'/3)) \mu}.
 \]
 The set $B$ can be chosen in ${N \choose K} \leq N^K$ ways. The set $A$ can be chosen in ${M \choose M/D} \leq (eD)^{M/D}$ ways.  It follows that the probability of ${\cal S}$ is bounded by
 \[
 N^K \cdot (eD)^{M/D} \cdot e^{-\Delta' (\ln (\Delta'/3)) \cdot K \cdot N_1 \cdot (1/D)},
 \]
 which, taking into account the relations between parameters, is less than $1$.~\qed

\end{document}